\documentclass[12pt]{elsarticle}
\usepackage[utf8]{inputenc}
\usepackage[T1]{fontenc}
\usepackage[english]{babel}
\usepackage[T1]{fontenc}
\usepackage{lmodern}
\usepackage{colortbl}
\usepackage{lastpage}
\usepackage{amscd}
\usepackage{algpseudocode}
\usepackage{amsmath}
\usepackage{multirow}
\usepackage{amsfonts}
\usepackage{amssymb}
\usepackage{latexsym}
\usepackage{mathrsfs}
\usepackage{float}
\usepackage{wrapfig}
\usepackage{color,framed}
\usepackage{longtable}
\usepackage{graphics}
\usepackage{graphicx}
\usepackage{amsmath}
\usepackage{acronym}
\usepackage{comment}

\usepackage{listings}
\newacro{TMN}{This Means Nothing}
\usepackage{mathdots}
\usepackage[left=3cm,right=3cm]{geometry}
\usepackage{ulem} 
\usepackage[colorlinks=true,linkcolor=blue,urlcolor=violet,bookmarks=true]{hyperref} 
\usepackage{array}
\newcolumntype{L}[1]{>{\raggedright\let\newline\\\arraybackslash\hspace{0pt}}m{#1}}
\newcolumntype{C}[1]{>{\centering\let\newline\\\arraybackslash\hspace{0pt}}m{#1}}
\newcolumntype{R}[1]{>{\raggedlef\left( t\let\newline\\\arraybackslash\hspace{0pt}}m{#1}}

\usepackage{color,xcolor,soul}

\def\dsum #1#2{\displaystyle{\sum_{#1}^{#2}}}

\newcommand{\Fq}{\mathbb{F}_{_q}}

\DeclareMathOperator{\Z}{Z}

\DeclareMathOperator{\aut}{Aut}

\DeclareMathOperator{\id}{Id}

\DeclareMathOperator{\Fix}{Fix}

\usepackage{ntheorem}
\theoremstyle{break}
{\theorembodyfont{\upshape}
	\newtheorem{definition}{Definition}
	\newtheorem{proposition}{Proposition} 
	\newtheorem{corollary}{Corollary}
	\newtheorem*{proof}{{Proof}}
	\newtheorem{example}{Example}
	
	\newtheorem{theorem}{Theorem}

	\newtheorem{remark}{Remark}
}	
\title{\LARGE\bfseries   $(n,\sigma)-$equivalence relation between skew constacyclic codes}

\begin{document}

\begin{frontmatter}

\author[nm]{ Hassan Ou-azzou  \corref{cor1}}
        		\ead{hassan.ouazzou@edu.umi.ac.ma} 
        		\author[nm]{Mustapha Najmeddine \corref{cor2}}
        		\ead{m.najmeddine@umi.ac.ma}
                    \author[na]{Nuh Aydin\corref{cor3}}
        		\ead{aydinn@kenyon.edu}

                    \cortext[cor1]{Principal corresponding author}
        		\address[nm]{ Department of Mathematics, ENSAM--Meknes, Moulay Ismail University, Morocco}
        		\address[na]{ Department of Mathematics, Kenyon College, Gambier, OH 43022, USA}

\begin{abstract}
In this paper we generalize the notion of $n$-equivalence  relation introduced by Chen et al. in \cite{Chen2014} to classify constacyclic codes of length $n$ over a finite field $\Fq,$ where $q=p^r$ is a prime power, to the case of skew  constacyclic codes without derivation. We call this relation   $(n,\sigma)$-equivalence  relation, where $n$ is the length of the code and $ \sigma$ is  an automorphism of the finite field. We compute  the number of  $(n,\sigma)$-equivalence classes, and we give conditions on $ \lambda$ and $\mu$  for which $(\sigma, \lambda)$-constacyclic codes and  $(\sigma, \lambda)$-constacyclic codes are equivalent with respect to   our $(n,\sigma)$-equivalence  relation. Under some conditions on $n$ and $q$ we prove that skew constacyclic codes are equivalent to  cyclic codes.  We  also prove that  when $q$ is even and $\sigma$ is the Frobenius autmorphism, skew constacyclic codes of length $n$ are equivalent to cyclic codes when  $\gcd(n,r)=1$. Finally we give some examples as  applications of the theory developed here.
\end{abstract}
        	\end{frontmatter}
\section{Introduction}

Coding theory plays a fundamental role in various applications, such as error detection and correction, data transmission, data storage and reliable communication. It involves the study of efficient encoding and decoding methods for transmitting data reliably over noisy channels. Cyclic codes are one of the most important families of linear codes for both theoretical and practical reasons. They establish a key link between coding theory and algebra, and their structure makes them convenient for implementation. Many of the best-known codes are either cyclic or related to cyclic codes. Cyclic codes were introduced in the 1950s by Prange in \cite{Prange1957} as linear codes with the property that the cyclic shift of any codeword is another codeword. One of the important generalizations of cyclic codes is the class of constacyclic codes, first introduced in \cite{Berlekamp1968}. Like cyclic codes, constacyclic codes also are also conveniently implemented  using simple shift registers. Their rich algebraic structure makes them a major subject of study for theoretical reasons. There are useful generalizations of cyclic codes such as quasi-cyclic (QC) and quasi-twisted (QT) codes that have been used to obtain a large number of new linear codes with better parameters e.g., \cite{Nuh2013, Aydin2019, Aydin2001, Almendras2018, Ou-azzou2023}. One of the generalizations of cyclic codes is skew-cyclic codes, which were introduced in \cite{Boucher2007}. Their study necessitates the use of the Ore ring $R:=\Fq[x,\sigma]$, where $\sigma$ is an automorphism on $\Fq$, and the multiplication is defined by the extension of the rule $x \cdot a = \sigma(a) x$ for all $a \in \Fq$. 

For a nonzero element $\lambda$ of $\Fq$, $\lambda$-constacyclic codes of length $n$ over $\Fq$ are characterized as the ideals $\langle g(x) \rangle$ of the quotient ring $\mathbb{F}_q[x] / \langle x^n - \lambda \rangle$ where the generator polynomial $g(x)$ is the unique monic polynomial of minimum degree in the code, and it is a divisor of $x^n - \lambda$. Cyclic codes ($\lambda=1$) are a particular case of constacyclic codes.  Although the quotient space $R /R \cdot \langle x^n - \lambda \rangle$ is not always a ring, hence skew cyclic codes are not necessarily  ideals, they  can be viewed as left $R$-submodules $\langle g(x) \rangle$ of the left $R$-module $R /R \cdot \langle x^n - \lambda \rangle $ where $g(x)$ is a right divisor of $x^n - \lambda$. Other characterizations of skew cyclic and skew constacyclic codes can be found in \cite{Gluesing2021, Almendras2018, Boucher2007, Boucher2009}.

In a previous work \cite{Chen2014}, Chen et al. introduced an equivalence relation ``$\sim_n$'' called $n$-equivalence for the nonzero elements of $\mathbb{F}_q$ to classify constacyclic codes of length $n$ over $\mathbb{F}_q$ such that the constacyclic codes belonging to the same equivalence  class have the same minimum distance. For $ \lambda, \mu \in \Fq^*$,  $\lambda \sim_n \mu$ means   the existence of  a nonzero scalar $ a \in \Fq^*$  such that the map $\psi_{a}$ from the ring $\mathbb{F}_q[x] / \langle x^n - \mu \rangle$ to  $\mathbb{F}_q[x] / \langle x^n - \lambda \rangle$, defined by $ \psi_a(f(x))= f(ax)$ is an isometry with respect to the Hamming distance. Equivalently,   $\lambda \sim_n \mu$ if the polynomial $ \lambda x^n-\mu$ has at least one root in $\Fq[x].$  So,  it is  easy to relate the generator polynomial of a $\lambda$-constacyclic code $C=\langle g(x)\rangle$ with that of $\psi_a(C)= \langle g(ax)\rangle$.
Some papers of interest about the equivalence of constacyclic codes can be found in \cite{Aydin2020, Aydin2021, Chen2015, Chen2012, Dastbasteh2023}.

In this paper we generalize the notion of $n$-equivalence  relation introduced by Chen et al. in \cite{Chen2014} to classify constacyclic codes of length $n$ over finite field $\Fq$ to the case of skew  constacyclic codes without derivation. We call this relation   $(n,\sigma)$-equivalence  relation, where $n$ is the length of the code and $ \sigma$ is  an automorphism of the finite field. We compute  the number of  $(n,\sigma)$-equivalence classes, and we give conditions on $\lambda$ and $\mu$  for which $(\sigma, \lambda)$-constacyclic codes and  $(\sigma, \lambda)$-constacyclic codes are equivalent with respect to the $(n,\sigma)$-equivalence  relation. Under some conditions on $n$ and $q$ we prove that skew constacyclic codes are equivalent to  cyclic codes.  We  also prove that  when $q$ is even and $\sigma$ is the Frobenuis autmorphism, skew constacyclic codes of length $n$ are equivalent to cyclic codes if $ \gcd(n,r)=1$. Finally we give some examples as applications of the theory developed here.

The rest of this paper is organized as follows: Section \ref{S2} provides a review of the basic background on skew polynomial rings and skew constacyclic codes. In section \ref{S3}, we study the properties of the equivalence relation $(n,\sigma)$-equivalence. Section \ref{S4} investigates the equivalence between skew constacyclic codes over $\mathbb{F}_{2^r}$ using the properties of the $(n,\sigma)$-equivalence  relation. Finally, in section \ref{S5} we present some examples as applications of the theory developed here.

\section{Skew $(\sigma, \lambda)-$constacyclic codes} \label{S2}

 Let $\Fq$ be a finite field with $q$ elements, where $q=p^{r}$ for some prime number $p$ and $r \in \mathbb{Z}_{\geq 1}$, and let $\sigma : \Fq \rightarrow \Fq$ be an automorphism of $\Fq$. Recall that the \textit{Frobenius automorphism of $\Fq$} is defined as the automorphism $\theta : \Fq \rightarrow \Fq$ that maps $a$ to $a^{p}$ and  generates the Galois group $\aut(\Fq).$  Then each   automorphism  $\sigma \in \aut(\Fq)$ can be expressed as a power of $\theta$, which means there exists an integer $s$, $ 0\leq s < r, $ such that $\sigma(a)=a^{p^s} $ for all $a\in \Fq.$ It is important to note also  that $\theta $ keeps $\mathbb{F}_{p}$ fixed i.e., $\theta(\mathbb{F}_{p})= \mathbb{F}_{p}$, and generates the cyclic  group of automorphisms of $ \Fq$.  

  \begin{definition}[Skew  $(\sigma,\lambda)$-constacyclic code]
  Let $\lambda \in \Fq^*$   and  $\sigma \in \aut(\Fq)$  be an automorphism of $\Fq.$  A linear code $C\subseteq \Fq^n$ is called a skew  $(\sigma,\lambda)$-constacyclic code of length $n$ over $\Fq$ if for each codeword  $(c_0,c_1,\ldots, c_{n-1})$ of $C,$  we have 
$$
\left( \lambda  \sigma(c_{n-1}), \sigma(c_0), \sigma(c_1), \ldots, \sigma(c_{n-2})\right) \in C.
$$
For  $ \lambda=\pm 1$ we obtain  skew $\sigma$-cyclic codes and $ \sigma$-negacyclic codes as special cases. 
 \end{definition}
 
 Let us recall that  a  map $T: \Fq^{^n} \rightarrow \Fq^{^n}$ is called  $ \sigma $-semi linear transformation if it is an additive map satisfying   $ T(\alpha u)=\sigma( \alpha ) T(u)$ for $\alpha \in \Fq$ and  $u\in \Fq^n.$ In the particular case  $\sigma=\id, \ T$ is a linear transformation. We can obtain some important classes of codes from semi linear transformations. 
 
 \begin{remark}\label{remark1}
  Let $\lambda \in \Fq^*$   and  $\sigma \in \aut(\Fq)$  be an automorphism of $\Fq.$
 \begin{enumerate}
 \item  Skew $(\sigma,\lambda)$-constacyclic codes are invariant under the $\sigma-$semi linear transformation $ T_{\sigma,\lambda},$  called the $(\sigma,\lambda)-$constacyclic shift, defined by:
 $$ T_{\sigma,\lambda}(v_0,v_1,\ldots,v_{n-1} )=  \left( \lambda \sigma( v_{n-1}), \sigma(v_0), \sigma(v_1), \ldots, \sigma(v_{n-2})\right), \text{for all }\ v=(v_0,v_1,\ldots, v_{n-1}) \in \Fq^n.$$
 \item $\lambda$-constacyclic codes are invariant under the linear transformation $ T_{\lambda }$ defined by:
$$ T_{\lambda}(v_0,v_1,\ldots,v_{n-1} )=  \left(\lambda v_{n-1}, v_0, v_1, \ldots, v_{n-2}\right), \text{for all }\ v=(v_0,v_1,\ldots, v_{n-1}) \in \Fq^n.$$
 \item  When $\sigma= \id,$ we get  $ T_{\lambda}= T_{\sigma,\id}.$
 \end{enumerate}
 \end{remark}
To investigate the structure of skew  $(\sigma,\lambda)-$constacyclic codes we need to recall the notion of skew polynomials.
 A $\sigma$-skew polynomial ring over $\Fq$ is the set 
$$\Fq [x,\sigma]=\left\{a_{0}+a_{1} x+\ldots+a_{n-1} x^{n-1}: a_{i} \in \Fq, n \in \mathbb{Z}_{\geq 0} \right\}$$
endowed with the usual polynomial addition and the multiplication is defined by
$$ x a= \sigma(a) x .$$
In the following proposition, we collect some properties of the skew polynomial ring $\Fq [x,\sigma]$.

\begin{proposition} [\cite{Ore33},  p. 483-486] \label{PPP.1}
Let $s$ be  the order  of the automorphism $\sigma,$   and  denote by $ \Fix_{\sigma}(\Fq)$ the  subfield of $\Fq$ fixed by $\sigma$.
\begin{enumerate}
\item The ring $\Fq [x,\sigma]$  is, in general, a non-commutative ring unless $\sigma $ is the identity automorphism of $\Fq.$ 
\item   An element $f \in \Fq [x,\sigma] $ is central if and only if $f$ belongs to $ \Fix_{\sigma}(\Fq)[x^{s}]$ i.e.,  the centre of $ \Fq [x,\sigma]$   is $\Z(R)= \Fix_{\sigma}(\Fq)[x^{s}].$
\item Two-sided ideals of $\Fq [x,\sigma]$ are generated by elements of the form $ h(x^s) x^{m}$, where $m$ is an integer and $h(x)\in  \Fix_{\sigma}(\Fq) [x].$ 
\item The ring $\Fq [x,\sigma]$ is a right (resp. left)   Euclidean domain.
\end{enumerate}
\end{proposition}

\noindent  It is important to note that in the polynomial ring $R$, the evaluation of    skew polynomials  is different from the  evaluation of polynomials in the usual polynomial rings over  commutative rings.
\begin{definition} [\cite{Lam1988} ]\label{def.2}
For $a\in \Fq,$ and $f(x)\in R$,
  the right evaluation of $f(x)$ at $a$ is the element $f(a):=r$ such that $f(x)= q(x)(x-a)+r.$
\end{definition}
In \cite[Proposition 2.9]{Lam1988}, Lam and Leroy  proved that, for each $\alpha\in \Fq$ we have   
\begin{equation}\label{e1}
f(\alpha)= a_0 + a_1 N_1(\alpha) +\ldots+ a_{n-1} N_{n-1}(\alpha)=  \dsum{i=0}{n} a_i N_i(\alpha)
\end{equation}
where $  N_{0}(\alpha)=1 ,$ and  for  each $i$ in $\mathbb{N}^*$  
$$
N_i(\alpha)=\sigma^{i-1}(\alpha) \ldots \sigma(\alpha) \alpha.
$$
By taking  $\sigma(a)=a^{p^s} $ for all $a\in \Fq,$ we obtain  for each $ i \in \mathbb{N}^*$  that
\begin{equation}\label{Eqsigma}
N_i(\alpha)=\alpha^{p^{s (i-1)}} \alpha^{p^{s (i-2)}} \ldots  \alpha^{p^{s}} \alpha =  \alpha^{ \sum_{j=0}^{i-1} p^{sj}}= \alpha^{{\frac{ p^{si}-1}{p^s-1}}}.
\end{equation}
Set $[i]_{s}=  \dfrac{ p^{si}-1}{p^s-1}$. Then the  right evaluation of $f(x)=\dsum{i=0}{n} a_i x^i$ at $\alpha\in \Fq$ is given by 
\begin{equation}
f(\alpha)= \dsum{i=0}{n} a_i \alpha^{[i]_s}.
\end{equation}

\noindent Now we recall the following characterization of skew $(\sigma,\lambda)-$constacyclic  codes.
 \begin{theorem}[ Theorem 3, Theorem 4 in  \cite{Almendras2018}]\label{T2.3}
 A linear code  $ C\subseteq \Fq^{^n} $ is a skew $(\sigma,\lambda)-$constacyclic  code if and only if it is  a left ideal of $ \Fq[x,\sigma]/\langle x^n-\lambda\rangle.$  Moreover,
 \begin{enumerate}
 \item There is a monic polynomial of least degree $ g(x)\in \Fq[x,\sigma] $ such that $ g(x) $  right divides $ x^n-\lambda $
 and $ C=\langle g(x)\rangle. $
 \item The set $ \{ g(x),x g(x),\ldots,x^{k-1}g(x)\} $ forms a basis of $ C $ and the dimension of $ C $ is $ k=n-\deg(g). $

 \item A generator matrix $ G $ of $ C $ is given by:
 \begin{equation}\label{eq10}
 G= \left(
 \begin{array}{cccccccc}
 g_{_0} &g_{_1} &\cdots&g_{_{n-k}}& 0 &\cdots &\cdots & 0\\
 0 & \sigma(g_{_0}) & \sigma(g_{_1}) &\cdots & \sigma(g_{_{n-k}}) & 0 &\cdots & 0\\
 \vdots &\ddots &\ddots &\ddots & &\ddots & &\vdots\\
 \vdots & &\ddots &\ddots &\ddots & &\ddots &\vdots\\
 0 &\ldots & &0 & \sigma^{k-1}(g_{_0}) & \sigma^{k-1}(g_{_1}) &\ldots & \sigma^{k-1}(g_{_{n-k}})\\
 \end{array}
 \right) 
 \end{equation}
 where $k=n-\deg(g) $ and $ g(x)= \displaystyle\sum_{i=0}^{n-k} g_ix^i. $
\item There is a one-to-one correspondence between  right divisors of $x^n-\lambda$ and skew $(\sigma,\lambda)-$constacyclic  codes of length $n$ over $\Fq$.
\end{enumerate}
 \end{theorem}
\begin{definition}[ Section 1.6, and 1.7 in \cite{Huffman2003}] Let $C_1$ and $C_2$ be two linear codes of length $n$ over $\Fq$. Then $C_1$ and $C_2$ are called \textbf{ equivalent} if there exists an $n \times n$ monomial matrix matrix $P$ such that $C_1=C_2 P$.
\end{definition} 

\begin{remark}
Two equivalent codes have the same minimum Hamming distance and in general same parameters. 
\end{remark}
In the following theorem  we show that in some cases skew $(\sigma,\lambda)-$constacyclic  codes are the same as ordinary $\lambda$-constacyclic codes.

\begin{theorem}\label{cyclic}
Let $\sigma $ be an automorphism of $\Fq$  of order $m$ and $C$ be a skew $(\sigma,\lambda)$-constacyclic  code of length $n$  over $\Fq$.  If $\gcd(m, n)=1$ and $\lambda\in \Fix_{\sigma}(\Fq)$ then $C$ is a $\lambda$-constacyclic code of length $n$.
\end{theorem}
\begin{proof}  Let $C=\langle g(x)\rangle$ be a skew $(\sigma,\lambda)$-constacyclic  code of length $n$ such that $\gcd(m, n)=1$. As shown in  \cite[Theorem 16]{Siap2011}, we know that there exist integers $\alpha_1, \alpha_2$ such that
 $\alpha_1 m=1+\alpha_2 n$, where $\alpha_2>0$. 
 
 Let $c(x)=c_0+c_1 x+\ldots+c_{n-1} x^{n-1} \in C$. We need to show that 
 $$ \lambda  c_{n-1}  +c_0 x+\ldots+c_{n-2} x^{n-1}  \in C.$$
Since $C$ is a skew $(\sigma,\lambda)$-constacyclic  code,  $x^{\alpha_1 m} c(x) \in C$ and   we have 
 $$
\begin{aligned}
  x^{\alpha_1 m} c(x) & = x^{1+\alpha_2 n } (c_0+c_1 x+\ldots+c_{n-1} x^{n-1}) \\
& =   \sigma^{1+\alpha_2 n}(c_0 ) x^{1+\alpha_2 n} +   \sigma^{1+\alpha_2 n}(c_1 ) x^{2+\alpha_2 n}  + ...+  \sigma^{1+\alpha_2 n}(c_{n-1} ) x^{n+\alpha_2 n}\\
&  =  \sigma^{\alpha_1 m}(c_0 ) x \lambda^{\alpha_2} +   \sigma^{m \alpha_1}(c_1 ) x^{2} \lambda^{\alpha_2}  + ...+ \sigma^{\alpha_1 m}(c_{n-1} ) \lambda^{1+\alpha_2}\\
&   =  (\lambda c_{n-1} + c_0  x +  c_1  x^{2}  + ...+ c_{n-2} x^{n-1}) \lambda^{\alpha_2}\\
&   =   \lambda^{\alpha_2} (\lambda c_{n-1} + c_0  x +  c_1  x^{2}  + ...+ c_{n-2} x^{n-1}) \in C \ \ \text{since} \ \  \lambda \in \Fix_{\sigma}(\Fq). \\
\end{aligned}
$$
Therefore, $ \lambda c_{n-1} + c_0  x +  c_1  x^{2}  + ...+ c_{n-2} x^{n-1} \in C , $ and hence $C$ is a $\lambda$-constacyclic code.
\qed \end{proof} 
 
\begin{theorem}  Let $\sigma $ be an automorphism of $\Fq$  of order $m$ and  $C$ be a skew $(\sigma,\lambda)$-constacyclic  code of length $n$ over $\Fq$. If $\gcd(m, n)=d$ and $\lambda\in \Fix_{\sigma}(\Fq)$ then $C$ is a  quasi-twisted code of length $n$ and index $d$ with shift constant $\lambda$.
\end{theorem}
\begin{proof}  Let $C=\langle g(x)\rangle$ be a skew $(\sigma,\lambda)-$constacyclic  code of length $n$ over $\Fq$ such that $\gcd(m, n)=d$. We know that there exist integers $\alpha_1, \alpha_2$ such that
 $\alpha_1 m=d+\alpha_2 n$, where $\alpha_2>0$.

 Let  $n= d s$ and 
 $$c	= ( c_{_{0,0}},c_{_{0,1}}, \ldots,c_{_{0,d-1}},	 c_{_{1,0}},c_{_{1,1}},\ldots, c_{_{1,d-1}} ,\ldots, c_{_{s-1,0}}, c_{_{s-1,1}},\ldots, c_{_{s-1,d-1}}) \in C.$$
We will show that $$ T_{\lambda}^{d}(c) = ( \lambda c_{_{s-1,0}},\lambda c_{_{s-1,1}},\ldots,\lambda c_{_{s-1,d-1}}, c_{_{0,0}},c_{_{0,1}}, \ldots, c_{_{1,d-1}} ,\ldots, c_{_{s-2,0}}, c_{_{s-1,1}},\ldots, c_{_{s-2,d-1}}) $$
where  $ T_{\lambda}$ is the $  \lambda$-constacyclic shift defined in Remark \ref{remark1}. Let    $T_{\lambda,\sigma}$ be the $\sigma$-semi-linear transformation defined by 
$$  T_{\lambda,\sigma} (c_0, c_1, \ldots, c_{n-1}) = ( \lambda \sigma(c_{n-1}), \sigma(c_{0}), \ldots, \sigma(c_{n-2})) = T_{\lambda} ( \sigma(c_{n-1}), \sigma(c_{0}), \ldots, \sigma(c_{n-2})) .$$  We know that  $C$ is  a skew $ \lambda$-constacyclic code if $  T_{\lambda,\sigma}(c) \in C$ for each $c\in C$, and  for each positive integer $k$ we have  $ T_{\lambda,\sigma}^k (c) =  T_{\lambda}^k (\sigma^k(c)) .$ It follows that 
 \begin{equation}
 T_{\lambda,\sigma}^{d +\alpha_2 n}(c)=    T_{\lambda}^{ d +\alpha_2 n} (\sigma^{ d +\alpha_2 n}(c))= T_{\lambda}^{ d +\alpha_2 n} (\sigma^{\alpha_1 m}(c))=T_{\lambda}^{ d +\alpha_2 n} (c)=   \lambda^{\alpha_2} T_{\lambda}^{ d }(c) \in C.
\end{equation}
Hence,   $ T_{\lambda}^{ d }(c) \in C $ and so  we obtain the  desired result.
\qed \end{proof} 



\section{ Equivalence between  skew constacyclic codes}  \label{S3}

In \cite{Chen2014} Chen et al.   introduced the notion  of $n$-equivalence  relation  on $\Fq^*$ for the purpose of  classifying  ordinary constacyclic codes of length $n$ over $\Fq.$ In this section  we generalize this  notion to the case of skew constacyclic codes  by defining  an equivalence relation on  $\Fq^*$ called $(n,\sigma)$-equivalence relation, where   $\sigma \in \aut(\Fq)$ is an automorphism of $\Fq.$ 

 \begin{definition} [$(\sigma,n)-$equivalence relation]
Let $n$ be  a positive integer. We say that  $\lambda$ and $\mu$  are  $(n,\sigma)$-equivalent in $\mathbb{F}_q^*$, denoted by $\lambda \sim_{n,\sigma} \mu$, if the skew polynomial $\lambda x^n-\mu \in \Fq[x,\sigma]$ has at least one root in $\Fq^*$, i.e.,  there is $\alpha\in \Fq^{*}$ such that $ \lambda  N_{n}(\alpha)  = \mu.$
\end{definition}

\begin{proposition}
The relation $\sim_{n,\sigma}$ is an equivalence relation on $\Fq^*$.
\end{proposition}
\begin{proof}
\begin{enumerate}
\item  $\lambda \sim_{n,\sigma} \lambda $ since $ 1$ is a  root of $\lambda x^n-\lambda$  in $\Fq$.
\item $\lambda \sim_{n,\sigma} \mu \iff$  there exists $\alpha \in \Fq$ such that $\alpha$ is a root of the polynomial $f(x)=\lambda x^n-\mu$. By equation (1), this is equivalent to    $$\lambda N_n(\alpha)= \mu  \iff  \lambda = \mu  N_n(\alpha^{-1}) \iff \mu \sim_{n,\sigma} \lambda. $$  
\item  If $\lambda \sim_{n,\sigma} \mu$ and $\mu \sim_{n,\sigma} \xi$ then there are $\alpha , \beta $ in $\Fq^*$ such that 
$$\lambda N_n(\alpha)= \mu   \ \ \text{and} \ \  \mu N_n(\beta)= \xi $$
Therefore  $$\lambda N_n(\alpha)=  N_n(\beta^{-1})  \xi $$
and so 
   $$\lambda N_n(\alpha)  N_n(\beta) =  \lambda N_n(\alpha \beta)  = \xi.$$  
Hence $ \lambda \sim_{n,\sigma} \xi.$ 
\end{enumerate}
\qed \end{proof}

The following theorem gives  characterizations of $(n,\sigma)$-equivalence  between two elements  $\lambda$ and $\mu$ of $\Fq^{*}.$
 
\begin{theorem}\label{Th.1}
Let   $\sigma \in \aut(\Fq)$ be an automorphism of $\Fq$  such that  $\sigma(a)=a^{p^s}$ for all $a\in \Fq$,     and  let $\lambda, \mu \in \mathbb{F}_q^*$. The following four statements are equivalent:
\begin{enumerate}
\item $\lambda \sim_{n,\sigma} \mu$ 
\item  There is an element $\alpha$ in $\mathbb{F}_q^*$ such that $ \lambda N_n(\alpha) =\mu $  and the following map 


\begin{equation*}
  	\begin{array}{cccc}
  	\varphi_{\alpha}:& \mathbb{F}_q[x,\sigma] /\langle x^n-\mu\rangle  &\longrightarrow & \mathbb{F}_q[x,\sigma] /\langle x^n-\lambda\rangle, \\ & & & \\
  	& f(x)& \longmapsto &  f(\alpha x)  
  	\end{array}
  	\end{equation*} 
is an $\mathbb{F}_q$-algebra isometry  with respect to the Hamming distance.

\item $\lambda^{-1} \mu \in \langle N_n(\xi) \rangle $, where  $\xi$ is a primitive element of $\Fq.$ 

\item $\lambda^{-1} \mu \in \langle \xi^{[n]_s} \rangle $, where  $\xi$ is a primitive element of $\Fq.$ 
\item $\left(\lambda^{-1} \mu\right)^d=1$, where $d=\dfrac{q-1}{\operatorname{gcd}([n]_s, q-1)}$.
\end{enumerate}
In particular, the number of $(n,\sigma)$-equivalence classes in $\mathbb{F}_q^*$ is equal to $\operatorname{gcd}([n]_s, q-1)$.
\end{theorem}
\begin{proof}
    \begin{enumerate}
       \item (1) $\Rightarrow$ (2). 
       Suppose that $ \lambda \sim_{n,\sigma} \mu $ then there is $ \alpha \in \Fq$ such that $ \lambda N_n(\alpha)=\mu$. We will  prove that $ \varphi_{\alpha}$ is an $\Fq$-algebra isometry. Let us first show that  $ \varphi_{\alpha}$ is an $ \Fq$-algebra isomorphism. Let $ f(x)=\dsum{i=0}{n-1} a_ix^i$ and $g(x)=  \dsum{i=0}{n-1} b_ix^i$  in $ \Fq[x,\sigma]/\langle x^n-\mu \rangle.$ On the one hand we have 

           $$
           \begin{array}{rl}
f(x)g(x) &=\dsum{j=0}{n-1}\dsum{i=0}{j} a_i \sigma^i\left(b_{j-i}\right)  x^j+ 
\dsum{j=0}{n-1} \dsum{i=j+1}{n-1} a_i \sigma^i\left(b_{n-i+j}\right)  x^{j+n} \\
& = \dsum{j=0}{n-1}\left(\dsum{i=0}{j} a_i \sigma^i\left(b_{j-i}\right) +  \dsum{i=j+1}{n-1} a_i  \sigma^i(b_{n-i+j}) \sigma^j(\mu) \right) x^j,  \\ 
&  
(\text{ since  $x^{j+n}= \sigma^j(\mu) x^j (\mod x^n-\mu) .$  })

\\
\end{array}
$$

As $\varphi_\alpha\left(x^j\right)=N_j(\alpha) x^j$, we obtain 
$$
\begin{array}{rl}
\varphi_\alpha(f(x)g(x)) & =\dsum{j=0}{n-1}\left(\dsum{i=0}{j} a_i \sigma^i\left(b_{j-i}\right)  + \dsum{i=j+1}{n-1} a_i \sigma^i(b_{n-i+j}) \sigma^j(\mu) \right) N_j(\alpha) x^j \\
& \\

\end{array}
$$
On the other hand we have:
$$
\begin{array}{rl}
   \varphi_\alpha(f(x))\varphi_\alpha(g(x))=  &   \dsum{j=0}{n-1}\left(\dsum{i=0}{j} a_i N_i(\alpha) \sigma^i\left(b_{j-i} N_{j-i}(\alpha)\right)\right) x^j+\\
 &  \dsum{j=0}{n-1}\left(\dsum{i=j+1}{n-1} a_i N_i(\alpha) \sigma^i\left(b_{n-i+j} N_{n+j-i}(\alpha)\right)\right) x^{j+n} . \\
\end{array}
$$

As $x^{j+n}=\sigma^j(\lambda) x^j (\mod x^n-\lambda)$, we have 
$$
\begin{array}{rl}
\varphi_\alpha(f(x))\varphi_\alpha(g(x)) & 
=\dsum{j=0}{n-1}\left(\dsum{i=0}{j} a_i \sigma^i\left(b_{j-i}\right) N_i(\alpha) \sigma^i\left(N_{j-i}(\alpha)\right) 
 \right) x^j+\\ 

& \dsum{j=0}{n-1}\left(\dsum{i=j+1}{n-1} a_i \sigma^i\left(b_{n-i+j}\right) N_i(\alpha) \sigma^i\left(N_{n+j-i}(\alpha)\right) \sigma^j(\lambda)\right) x^j .\\
\end{array}
$$
 According to \cite[Proposition 2.1]{Cherchem2016} we  have  $N_i(\alpha) \sigma^i\left(N_{j-i}(\alpha)\right)=N_j(\alpha),$ and so
$$N_i(\alpha) \sigma^i\left(N_{n+j-i}(\alpha)\right) \sigma^j(\lambda)= N_{j+n}(\alpha) \frac{ \sigma^j(\mu)}{\left(\sigma^j\left(N_n(\alpha)\right)\right)}= \sigma^j(\mu) N_j(\alpha), \ \text{since }   \lambda  N_{n}(\alpha)  = \mu .
$$
It follows that 
$$
\begin{array}{rl}
  \varphi_\alpha(f(x))\varphi_\alpha(g(x))&
  =\dsum{j=0}{n-1}\left(\dsum{i=0}{j} a_i \sigma^i\left(b_{j-i}\right)+\dsum{i=j+1}{n-1} a_i \sigma^i\left(b_{n-i+j}\right) \sigma^j(\mu) \right) N_j(\alpha) x^j     \\
     &  =\varphi_\alpha(f(x)g(x)),
\end{array}
$$
which means that $ \varphi_{\alpha}$ is a morphism of an $\Fq$-algebra.
           
 Finally we observe that for any $f(x)\in \mathbb{F}_q[x,\sigma] /\langle x^n-\mu \rangle$, $f(x) $ and $ f(\alpha x)$ have the same weight. Therefore, $ \varphi_{\alpha}$ preserves the Hamming weight and so it is an $\Fq$-algebra isometry with respect to Hamming distance.

        \item (2) $\Rightarrow$ (3). Suppose there exists an $\alpha \in \mathbb{F}_q^*$ such that the map 
        $$\varphi: \mathbb{F}_q[x,\sigma] /\langle x^n-\mu\rangle \to \mathbb{F}_q[x,\sigma] /\langle x^n-\lambda\rangle, \ \ f(x) \mapsto f(\alpha x)$$
        is an  $\mathbb{F}_q$-algebra isometry  with respect to the Hamming distance. From the fact that $ \mu=\varphi_{\alpha}(\mu) $ and $ x^n= \lambda \mod (x^n-\lambda)$  we obtain 
        
$$
\mu=\varphi_{\alpha}(\mu)= \varphi_{\alpha}\left(x^n\right)=\varphi_{\alpha}(x)^n=
(\alpha x)^n=N_{n}(\alpha) x^n= N_{n}(\alpha) \lambda.
$$
This gives  $ \lambda^{-1} \mu = N_{n}(\alpha). $  Let $\xi $ be a primitive element of $ \Fq$ then $\alpha= \xi^i$ for some $ 0\leq i\leq q-1.$ It follows that 

$$ \lambda^{-1} \mu = N_{n}(\xi^i)=  N_{n}(\xi)^i , \ \text{ which means that }  \ \lambda^{-1} \mu \in  \langle  N_{n}(\xi) \rangle. $$

\item (3) $\Rightarrow$ (4). Since $ N_n(\xi)= \xi^{[n]_s} $ we have $ \langle \xi^{[n]_s} \rangle= \langle N_n(\xi) \rangle$ and so $ \lambda^{-1} \mu \in \langle \xi^{[n]_s} \rangle. $
\item (4) $\Rightarrow$ (5).  Since $\mathbb{F}_q^*=\langle \xi \rangle$ is a cyclic group, $\langle\xi^{[n]_s}\rangle$ is the unique subgroup of order $d=\dfrac{q-1}{\operatorname{gcd}([n]_s, q-1)}$ and any subgroup with order dividing $d$ is contained in the subgroup $\langle\xi^{[n]_s}\rangle$. It follows that  $(\lambda \mu^{-1})^d =1 .$
    \item (5) $\Rightarrow$ (1). Let  $\lambda$ and $\mu$ in $\mathbb{F}_q^*$ such that $ (\lambda \mu^{-1})^d =1.$  That means $ \lambda \mu^{-1} $ belongs to the cyclic subgroup  of $ \Fq^{*}$ of order $d.$ However, the unique subgroup of  $ \Fq^{*}$ of order $d$ is  $ \langle \xi^{[n]_s} \rangle$. 
    It follows that    $\lambda^{-1}\mu=\xi^{i[n]_s} \in\langle \xi^{[n]_s}\rangle,$ for some   integer $0\leq i\leq q-1.$ Let $\alpha=  \xi^{i}, $ then $\lambda^{-1}\mu=\alpha^{[n]_s} = N_{n}(\alpha),$ and so $\lambda N_{n}(\alpha) =\mu $ which means that $ \lambda \sim_{n,\sigma} \mu.$
\end{enumerate}

Finally, by the equivalence of (1) and (5), the number of  $(n,\sigma)$-equivalence classes in $\mathbb{F}_q^*$ is equal to the index of  $\langle\xi^{[n]_s}\rangle$ in $\mathbb{F}_q^*$, which is given by 
$$
\left|\mathbb{F}_q^* \ : \ \langle\xi^{[n]_s}\rangle\right|=\frac{q-1}{\left|\langle\xi^{[n]_s}\rangle\right|}=\dfrac{q-1}{\frac{q-1}{\operatorname{gcd}([n]_s, q-1)}}=\operatorname{gcd}([n]_s, q-1) .
$$    
\qed \end{proof} 

\begin{corollary}\label{Cor.1}
    Let $\lambda, \mu\in \Fq^{*}$ be such that $ \lambda \sim_{n,\sigma} \mu $ i.e.,  $ \lambda N_{n}(\alpha)= \mu$ for some $ \alpha \in \Fq^{*}.$ Then the   skew $( \sigma,\lambda)$-constacyclic codes of length $n$ over $\Fq$ are equivalent to skew $(\sigma,\mu)$-constacyclic codes of length $n$ over $\Fq$ via the following $\Fq$-algebra isomorphism  $\varphi_{\alpha},$ which preserves the Hamming weight.

\begin{equation}
  	\begin{array}{cccc}
  	\varphi_{\alpha}:& \mathbb{F}_q[x,\sigma] /\langle x^n-\mu\rangle  &\longrightarrow & \mathbb{F}_q[x,\sigma] /\langle x^n-\lambda\rangle, \\ & & & \\
  	& f(x)& \longmapsto &  f(\alpha x)  
  	\end{array}
  	\end{equation} 
\end{corollary}

\begin{remark}[ Theorem 1, \cite{Boulanouar2021}]
The result of the above corollary generalizes the result  \cite[Theorem 1]{Boulanouar2021}, where  the authors proved  that the skew $(\sigma,\lambda)$-constacyclic  codes are equivalent to skew cyclic codes if there is $\alpha \in \Fq^*$ such that $ \lambda = N_n(\alpha^{-1})$,  and to  skew  negacyclic codes if $ \lambda = -N_n(\alpha^{-1})$. This  corresponds to the special case $\mu =1$ ( or $\mu=-1$ ) and   $ \lambda \sim_{n,\sigma} 1$  ( or $\lambda \sim_{n,\sigma} -1 $) in our result.  
\end{remark}

\begin{corollary} 
Let  $\sigma \in \aut(\Fq)$ be an automorphism of $\Fq$ such that  $\sigma(a)=a^{p^s}$ for all $a\in \Fq$, and let $n$ be a positive integer such that $ \gcd([n]_s, q-1)=1$. Then any two nonzero elements $\lambda$ and $\mu$ of $\mathbb{F}_q^*$ are $(n,\sigma)$-equivalent to each other, i.e., $\lambda N_n(\alpha)=\mu$ for an $ \alpha\in \mathbb{F}_q^*$, and the map $\varphi_{_\alpha}: \mathbb{F}_q[x,\sigma] /\langle x^n-\mu\rangle \rightarrow \mathbb{F}_q[x,\sigma] /\langle x^n-\lambda\rangle$ which maps $f(x)$ to $f(\alpha x)$ is an $\mathbb{F}_q$-algebra isometry with respect to the Hamming distance.
\end{corollary}
\begin{proof}
 Since $\operatorname{gcd}([n]_s, q-1)=1$, there is only one $(n,\sigma)$-equivalence class in $\mathbb{F}_q^*.$ In other words, for any $\lambda, \mu \in \mathbb{F}_q^*$, we have  $\lambda \sim_{n,\sigma} \mu$.
\qed \end{proof} 
 The following result was proved in  \cite{Boulanouar2021} by using a basic result from  field theory (see  \cite{Roman1995}) which says that  the equation $ x^n= \xi^i$ has a solution in $\Fq^{*}$ if and only if $ \gcd(n,q-1)$ divides $i.$ In our case we   obtain a similar result by using the fact that we have only one $(n,\sigma)$-equivalence class in $\mathbb{F}_q^*$ when $ \gcd([n]_s, q-1)=1.$
\begin{corollary} [Proposition 1, \cite{Boulanouar2021}]
Let  $\sigma \in \aut(\Fq)$ be an automorphism of $\Fq$ such that  $\sigma(a)=a^{p^s}$ for all $a\in \Fq,$ and  let $n$ be a positive integer such that $ \gcd([n]_s, q-1)=1$. Then all skew  $(\sigma,\lambda)$-constacyclic codes are equivalent to skew cyclic codes.
\end{corollary}

\begin{proposition}

Let  $\sigma \in \aut(\Fq)$ be an automorphism of $\Fq$ such that  $\sigma(a)=a^{p^s}$ for all $a\in \Fq,$ and let $n$ be a positive integer  $d= \gcd([n]_s, q-1).$ For each $\lambda \in \Fq^*$ such that $ \lambda^d=1,$  all skew  $(\sigma,\lambda)$-constacyclic codes are equivalent to skew cyclic codes.
\end{proposition}

\begin{proof}
 In the statement (5) of Theorem \ref{Th.1}, we take $ \mu=1.$ By the fact that $ \lambda^{d}= (\lambda^{-1})^d=1, $ the result holds.
\qed \end{proof}

\section{Skew constacyclic codes over $\mathbb{F}_{2^r}$}  \label{S4}

Let $q=2^r$ and  $\sigma\in \aut(\Fq)$ be  the \textbf{Frobenius  automorphism} of  $\mathbb{F}_{q},$ i.e., $ \sigma(a)=a^2$ for all $a\in \mathbb{F}_{q}.$   Note that the order of $\sigma$ is $r.$  In the following theorem we investigate equivalence among $(\sigma,\delta)$-constacyclic codes of length $n$ over $\mathbb{F}_{q}.$ 

\begin{theorem} \label{Thh.5}
    Let $ q=2^r$ and  $\sigma \in \aut(\Fq)$ be the Frobenuis automorphism of $\Fq.$  Let $C$ be a $(\sigma,\lambda)$-constacyclic code  of length  $ n$ over $\Fq.$ 
    The number of $(n,\sigma)$-equivalence  classes is equal  $ 2^{\gcd(n,r)}-1$,  and 
    \begin{enumerate}
        \item If $ \gcd(n, r)=1,$ then  $C$  is  equivalent to a  cyclic code of length $n$ over $\Fq.$
    
        \item  If $\gcd(n,r)= d\neq 1$, then  there is a unique $ 0\leq j\leq 2^d-2 $ such that $ C$ is equivalent to  a skew $(\sigma, \xi^j)$-constacyclic code, where $\xi$ is a primitive element of $\Fq.$ 
    \end{enumerate}
\end{theorem}

\begin{proof}
 We have    $ \gcd([n]_1, q-1)= \gcd(2^n-1, 2^r-1)= 2^{\gcd(n,r)}-1.$ Then according to Theorem \ref{Th.1}, the number of $ (n,\sigma)$-equivalence classes is  $2^{\gcd(n,r)}-1.$
 \begin{enumerate}
 \item  If $ n=1 \mod r$, then $\gcd([n]_1, q-1)=1 $ and so we have only one equivalence class which means that all  $(\sigma,\lambda)$-constacylic codes  of length  $ n$ over $\Fq$ are equivalent to skew cyclic codes of length $n$ over $\Fq.$ On the other hand the order of $\sigma$ is $r$ and $\gcd(n,2k)=1,$ then by Theorem \ref{cyclic}, $C$ is equivalent to a cyclic code  of length $n$  over $\Fq.$ 
 
        \item If  $\gcd(n,r)= d\neq 1,$ and  $\gcd([n]_1, q-1)=2^d-1, $
        we can decompose $\Fq^{*}$   into a disjoint union of cosets of the subgroup $\langle \xi^{2^d-1} \rangle $ as follows:
        $$  \Fq^{*}= \langle \xi \rangle= \langle N_n(\xi) \rangle \cup \xi \langle N_n(\xi) \rangle \cup \ldots \cup \xi^{2^d-2} \langle N_n(\xi) \rangle.  $$
        
       So the element  $\lambda \in \Fq^*$  belongs to exactly one of the cosets, i.e., there is a unique integer $ 0\leq j\leq 2^d-2 ,$ such that $ \lambda \in  \xi^j \langle N_n(\xi) \rangle$ and so $ \lambda$ is $(\sigma,n)$-equivalent to $\xi^j$. 
        
\end{enumerate}
\qed \end{proof}

When $q=2^r$ and $r$ is an even number we have the following result.

\begin{corollary}\label{Corr4}
   Let $q=2^{2k}$ and  let $\sigma$ be the Frobenuis automorphism of $\Fq.$  Then all   $(\sigma,\lambda)$-constacyclic codes of  odd length  $ n$ are 
   equivalent to cyclic codes of length $n$ over $\Fq.$
\end{corollary}
\begin{proof}
Let $C$ be a  $(\sigma,\lambda)$-constacyclic code of  odd length  $ n$ over $\Fq.$  On the one hand we have  $ \gcd([n]_1, 2k)= 2^{\gcd(n,2k)}-1= 2^1-1=1,$ hence $C$ is  equivalent to a skew cyclic code of length $n$ over $\Fq.$  On the other hand the order of $\sigma$ is $2k$ and $\gcd(n,2k)=1$. Then by Theorem \ref{cyclic}, $C$ is equivalent to a cyclic code  of length $n$  over $\Fq.$
\qed \end{proof} 
\begin{corollary}
      All  skew   $(\sigma,\lambda)$-constacylic codes of  odd length   $ n$ over $\mathbb{F}_{4}$ are 
   equivalent to cyclic codes.
\end{corollary}

\section{Examples}  \label{S5}
\noindent In this section, we give some examples of  applications of the theory developed here. 

\begin{example}[Skew constacyclic codes over $\mathbb{F}_4$]
Let $q = 4$, and $\xi$ be a primitive element of $\mathbb{F}_4$. In this example, we are interested in skew constacyclic codes of length $n$ over $\mathbb{F}_4$. Since the degree of the extension field $\mathbb{F}_4/\mathbb{F}_2$ is $2$, we have $\text{Aut}(\mathbb{F}_4) = \{\text{id}, \theta\}$, where $\theta$ is the Frobenius automorphism of $\mathbb{F}_4$, i.e., $\theta(a) = a^2$ for all $a\in \mathbb{F}_4$. According to Theorem \ref{Th.1}, the number of $(n, \theta)$-equivalence classes is equal to $\gcd([n]_1, q-1) = \gcd(2^n-1, 3)$. By Theorem \ref{Thh.5}, we have two cases:

\begin{enumerate}
    \item \textbf{If $n$ is odd,} we have $\gcd(2^n-1, 2^2-1) = 2^{\gcd(n,2)}-1 = 2-1 = 1$, and so there is only one $(n, \theta)$-equivalence class. Since $n$ is coprime with the order of $\theta$, all skew $(\theta, \lambda)$-constacyclic codes of odd length over $\mathbb{F}_4$ are $(n, \theta)$-equivalent to cyclic codes. This means that there exists an element $\alpha$ in $\mathbb{F}_4^*$ such that $\alpha^{[n]_1} \lambda = 1$. Studying skew $(\lambda, \theta)$-constacyclic codes, in this case, is reduced to studying cyclic codes of length $n$ over $\mathbb{F}_4$.

     \item \textbf{If $n$ is even}, we have $\gcd(2^n-1, 2^2-1) = 2^{\gcd(n,2)}-1 = 3$, then we have $3$ $(n, \theta)$-equivalence classes. Let $\lambda \in \mathbb{F}_{4}^*$, then by Theorem \ref{Thh.5}, $\lambda$ is $(n, \theta)$-equivalent to $1, \xi,$ or $\xi^2$. To construct skew $(\lambda, \theta)$-constacyclic codes, in this case, we need to study three classes, which are skew $\theta$-cyclic codes, skew $(\xi,\theta)$-constacyclic, and $(\xi^2,\theta)$-constacyclic codes.
\end{enumerate}
\end{example}

\begin{example}[Skew constacyclic codes over $\mathbb{F}_8$]
Let $q = 8$, and $\xi$ be a primitive element of $\mathbb{F}_8$. In this example, we are interested in skew constacyclic codes of length $n$ over $\mathbb{F}_8$. Since the degree of the extension field $\mathbb{F}_8/\mathbb{F}_2$ is $3$, we have $\text{Aut}(\mathbb{F}_8) = \{\text{id}, \theta, \theta^2\}$, where $\theta$ is the Frobenius automorphism of $\mathbb{F}_8$, i.e., $\theta(a) = a^2$ for all $a\in \mathbb{F}_8$. In this case, we need to study skew $(\theta, \lambda)$-constacyclic and skew $(\theta^2, \lambda)$-constacyclic codes of length $n$ over $\mathbb{F}_8$.

\begin{enumerate}
    \item \textbf{Case 1: Skew $(\theta, \lambda)$-constacyclic codes:}
According to Theorem \ref{Th.1}, the number of $(n, \theta)$-equivalence classes is equal to $\gcd([n]_1, q-1) = \gcd(2^n-1, 7)$. By Theorem \ref{Thh.5}, we have two cases:

\begin{enumerate}
    \item \textbf{If $\gcd(n,3)=1$,} we have $\gcd(2^n-1, 2^3-1) = 2^{\gcd(n,3)}-1 = 2-1 = 1$, and so there is only one $(n, \theta)$-equivalence class. Since $n$ is coprime to the order of $\theta$, the order of $\theta$ is $3$, then all skew $(\theta, \lambda)$-constacyclic codes of odd length over $\mathbb{F}_8$ are $(n, \theta)$-equivalent to cyclic codes. This means that there exists an element $\alpha$ in $\mathbb{F}_8^*$ such that $\alpha^{[n]_1} \lambda = 1$. Studying skew $(\lambda, \theta)$-constacyclic codes, in this case, is reduced to studying cyclic codes of length $n$ over $\mathbb{F}_8$.

    \item \textbf{If $\gcd(n,3)=3$}, we have $\gcd(2^n-1, 2^3-1) = 2^{\gcd(n,3)}-1 = 7$. Then we have $7$ $(n, \theta)$-equivalence classes. So, in this case, we need to consider all possible skew $(\theta, \lambda)$-constacyclic codes for any $\lambda\in \mathbb{F}_8^*$.
\end{enumerate}

\item \textbf{Case 2: Skew $(\theta^2, \lambda)$-constacyclic codes:}
According to Theorem \ref{Th.1}, the number of $(n, \theta^2)$-equivalence classes is equal to $\gcd([n]_2, q-1) = \gcd\left(\frac{2^{2n}-1}{2^2-1}, 7\right)$. In this case, we performed computations  using Magma software, where we calculated $\operatorname{gcd}([n]_2, q-1)$ for $n$ ranging from $1$ to $100,000,000$. We observed two distinct cases.

\begin{enumerate}
    \item \textbf{If $\gcd(n,3)=1$,} we have $\gcd([n]_2, 7) = 1$, and so there is only one $(n, \theta^2)$-equivalence class. Since $n$ is coprime to the order of $\theta^2$, the order of $\theta^2$ is $3$, then all skew $(\theta^2, \lambda)$-constacyclic codes of odd length over $\mathbb{F}_8$ are $(n, \theta^2)$-equivalent to cyclic codes. This means that there exists an element $\alpha$ in $\mathbb{F}_8^*$ such that $\alpha^{[n]_2} \lambda = 1$. Studying skew $(\lambda, \theta^2)$-constacyclic codes, in this case, is reduced to studying cyclic codes of length $n$ over $\mathbb{F}_8$.

    \item \textbf{If $\gcd(n,3)=3$}, we have $\gcd([n]_2, 7) = 7$. Then we have $7$ $(n, \theta^2)$-equivalence classes. So, in this case, we need to consider all possible skew $(\theta^2, \lambda)$-constacyclic codes for any $\lambda \in \mathbb{F}_8^*$.
\end{enumerate}
\end{enumerate}
\end{example}

\section*{ Conclusion}
In this paper, we introduced the    $(n,\sigma)$-equivalence  relation, where $n$ is the length of the code and $ \sigma$ is an automorphism of the finite field. We computed the number of  $(n,\sigma)$-equivalence  classes, and we give conditions on $ \lambda$ and $\mu$  for which $(\sigma, \lambda)$-constacyclic codes and  $(\sigma, \lambda)$-constacyclic codes to be equivalent using the properties of our $(n,\sigma)$-equivalent  relation. Under some conditions on $n$ and $q$, we proved that skew constacyclic codes are equivalent to cyclic codes. For future work, it would be interesting to introduce the concept of ``skew cyclotomic cosets'' to facilitate the factorization of $x^n-1$ and $x^n-a$ in $\mathbb{F}_q[x, \sigma].$ This may give an efficient way to classify skew constacyclic codes of length $n$ over a finite field $\Fq$.


\begin{thebibliography}{}

\bibitem{Aydin2019} N. Aydin, Some new linear codes from skew cyclic codes and computer algebra challenges, Appl. Algebra Eng. Commun. Comput. 30(3), 185-191 (2019).
\bibitem{Nuh2013} N. Aydin, J. M. Murphree,  New linear codes from constacyclic codes, Journal of the Franklin Institute 351, 1691–1699 (2014).
 \bibitem{Aydin2001} N. Aydin, I. Siap, D.K. Ray-Chaudhuri, The structure of $1$-generator quasi-twisted codes and new linear codes, Des. Codes Cryptogr. 24 (3),  313–326 (2001).
 \bibitem{Aydin2020} N. Aydin, J. Lambrinos,  O. VandenBerg, On equivalence of cyclic codes, generalization of a quasi-twisted search algorithm, and new linear codes, Designs, Codes and Cryptography 87, 2199-2212 (2019).
\bibitem{Aydin2021}  N. Aydin,  O. VandenBerg, A new algorithm for equivalence of cyclic codes and its applications,  Applicable Algebra in Engineering, Communication and Computing (2021): 1-13.
\bibitem{Almendras2018} V.  Almendras, A. Tironi, On the dual codes of skew constacyclic codes, Advances in Mathematics of Communications 12.4 (2018).
\bibitem{Boulanouar2021} R. Boulanouar, A. Batoul, D. Boucher, An Overview on Skew Constacyclic Codes and their Subclass of LCD Codes, Advances in Mathematics of Communications 15(4),  611-632 (2021).
\bibitem{Berlekamp1968} E.R. Berlekamp, Algebraic Coding Theory, Mc Graw-Hill Book Company, New York (1968).
\bibitem{Boucher2007} D. Boucher, W. Geiselmann, F. Ulmer, Skew-cyclic codes, Applicable Algebra in Engineering, Communication and Computing 18, 379-389 (2007).
\bibitem{Boucher2009} D. Boucher, F. Ulmer, Coding with skew polynomial rings, Journal of Symbolic Computation 44, 1644-1656 (2009).
\bibitem{Chen2015} B. Chen, H. Q. Dinh, H. Liu, Repeated-root constacyclic codes of length $2\ell^m p^n$, Finite Fields and Their Applications, Volume 33, Pages 137-159 (2015).
\bibitem{Chen2014} B. Chen, H. Q. Dinh, H. Liu, Repeated-root constacyclic codes of length $\ell p^s$ and their duals, Discrete Applied Mathematics 177, 60–70 (2014).
\bibitem{Chen2012} B. Chen, Y. Fan, L. Lin, H. Liu,  Constacyclic codes over finite fields, Finite Fields and Their Applications 18, 1217–1231 (2012).
 \bibitem{Cherchem2016} A. Cherchem, A. Leroy, Exponents of skew polynomials,  Finite Fields and Their Applications 37, 1–13 (2016).   
\bibitem{Dastbasteh2023} R. Dastbasteh, P. Lisoněk, On the equivalence of linear cyclic and constacyclic codes, Discrete Mathematics, Volume 346, Issue 9, (2023).

\bibitem{Gluesing2021} H. Gluesing-Luerssen, Introduction to skew-polynomial rings andskew-cyclic codes,  Concise Encyclopedia of Coding Theory, pp. 151–180, Chapman and Hall/CRC, 2021.
\bibitem{Huffman2003} W.C. Huffman, V. Pless, Fundamentals of error-correcting codes, Cambridge University Press (2003).
\bibitem{Lam1988} T.Y. Lam, A. Leroy, Algebraic conjugacy classes and skew polynomial ring, Perspectives in Ring Theory, Springer, 153–203 (1988).
\bibitem{Ou-azzou2023}  H. Ou-azzou, M. Najmeddine, N. Aydin, E. M. Mouloua, On the algebraic structure of $(M,\sigma,\delta)$-skew codes, Journal of Algebra, (2023).
 \bibitem{Ore33} O. Ore, Theory of non-commutative polynomials,  Annals Math., 34:480-508, (1933). 
\bibitem{Prange1957} E. Prange, Cyclic Error-correcting Codes in Two Symbols, Air Force Cambridge Research Center, (1957).
\bibitem{Roman1995}  S. Roman, Field Theory, Springer-Verlag, New York (1995).
 \bibitem{Siap2011} I. Siap, T. Abualrub, N. Aydin, Skew cyclic codes of arbitrary length, International Journal of Information and Coding Theory 2(1), 10-20 (2011).
\end{thebibliography}
\end{document}